%% file: root.tex
\documentclass{scrarticle}

\usepackage{silence}
\usepackage{xcolor}
\usepackage{algorithm}
\usepackage{algorithmic}
\usepackage[numbers]{natbib}
\usepackage[utf8]{inputenc}
\usepackage[english]{babel}
\usepackage{microtype}
\usepackage{graphicx}
\usepackage{subfigure}
\usepackage{booktabs}
\usepackage{hyperref}
\usepackage{amsmath}
\usepackage{amssymb}
\usepackage{mathtools}
\usepackage{amsthm}
\usepackage{xspace}
\usepackage[capitalize,noabbrev]{cleveref}

\hypersetup{
    colorlinks,
    linkcolor={red!80!black},
    citecolor={blue!80!black},
    urlcolor={blue!80!black}
}

\theoremstyle{plain}
\newtheorem{theorem}{Theorem}[section]
\newtheorem{proposition}[theorem]{Proposition}
\newtheorem{lemma}[theorem]{Lemma}

\theoremstyle{definition}
\newtheorem{definition}[theorem]{Definition}

\theoremstyle{remark}

\newtheorem{example}{Example}

\renewcommand{\phi}{\varphi}
\renewcommand{\epsilon}{\varepsilon}

\newcommand{\stonr}{STON’R\xspace}

\DeclareMathOperator{\Smaller}{Smaller}
\DeclarePairedDelimiterX\set[1]\lbrace\rbrace{#1}

\providecommand{\norm}[1]{\lVert#1\rVert}

\def\.#1{\mathbf{#1}}
\newcommand{\G}{\mathcal{G}}

\newcommand{\R}{\mathbb{R}}

\title{\stonr Converges to First-Order Nash~Equilibria of Multiplayer Games}

\author{Marika Kosohorská \and Tomáš Kroupa \and Tomáš Votroubek}
\publishers{Department of Computer Science,\\ Czech Technical University in Prague\thanks{\raggedright The authors acknowledge the support of the grant 
no. 24-12046S of the Czech Science Agency.\par}}

\begin{document}

\maketitle
\begin{abstract}
Nonconcave games present a unique challenge, as neither pure Nash equilibria nor local Nash equilibria (LNE) are guaranteed to exist, even in zero-sum settings. Additionally, computing approximate LNE in smooth multiplayer games over bounded regions is PPAD-hard. These challenges, coupled with the inherent complexity, have driven recent research toward broader equilibrium concepts, such as min-max critical points, and first-order Nash equilibria (FONE), which correspond to solutions of specific non-monotone variational inequalities. This paper addresses general-sum multiplayer games with compact convex strategy sets and smooth, nonconcave utility functions. Daskalakis et al. introduced the STON'R algorithm for solving variational inequality problems and established convergence under smoothness assumptions. They further showed that the algorithm’s limit points correspond to equilibria in specific classes of games, namely local minimax equilibria in two-player zero-sum games and Nash equilibria in smooth concave games. In this work, we extend the convergence result to multiplayer general-sum games and show that the variational inequality solutions targeted by STON'R correspond to first-order Nash equilibria (FONE), a general game-theoretic solution concept that unifies these previously studied cases. We demonstrate the effectiveness and robustness of the algorithm on various examples from recent literature.
\end{abstract}

\section{Introduction}
Computing the Nash equilibrium for continuous games is a significant challenge in game theory. It is known to be difficult even for two-player zero-sum games with compact strategy spaces when the loss function of the minimizing player is not convex-concave; i.e., it is not convex in the player's own variables or not concave in the opponent's variables. Although the existence of a mixed Nash equilibrium is guaranteed under the continuity of the loss function, a pure Nash equilibrium may not exist, and determining whether it does is NP-hard \cite{daskalakis2021complexity}.

Therefore, first-order solution concepts have been introduced, such as min-max critical points, and first-order Nash equilibria \cite{daskalakis2023stay,Tsaknakis2021-jv,Nouiehed2019-ez,razaviyayn2020nonconvex}. These points represent solutions to a specific non-monotone variational inequality and are guaranteed to exist for smooth functions. They correspond to the fixed points of the projected gradient descent-ascent dynamics, and hence their computational complexity lies in the PPAD class \cite{daskalakis2009complexity}. In a recent analysis of the special case of monotone variational inequalities associated with convex-concave minimax problems, \citet{sedlmayer2023fast} introduced a first-order method that uses a single operator evaluation and a single projection in each iteration.

The \stonr algorithm \cite{daskalakis2023stay} for finding min-max critical points is based on a topological argument, leveraging the equivalence between the fixed points of the projected gradient descent-ascent dynamics and min-max critical points. Another method exploits the regularized Nikaido-Isoda function \cite{Tsaknakis2021-jv} to find an approximate first-order Nash equilibrium within a neighborhood in nonconvex two-player games, where the players have access only to local stochastic gradients. A first-order numerical method for special nonconvex-concave games based on multi-step gradient descent-ascent was analyzed and its performance tested on adversarial learning problems \cite{Nouiehed2019-ez}. Continuous games without assumptions on convexity or concavity have recently found numerous applications in AI/ML, including signal and data processing, GAN training, and robust machine learning \cite{razaviyayn2020nonconvex}.

This paper explores nonzero-sum continuous games involving more than two players, where the strategy spaces are compact and convex, and the utility functions are not necessarily concave. We extend the \stonr algorithm, originally designed for nonconvex-nonconcave continuous two-player zero-sum games, to handle nonconcave nonzero-sum continuous games. This extension is feasible because the original \stonr algorithm essentially addresses non-monotone variational inequalities linked to Nash equilibrium problems. We validate the extended algorithm's soundness and effectiveness, demonstrating its performance through several examples from recent research literature.

In two-player zero-sum games, the solution concept for \stonr is the min-max critical point, which coincides with first-order Nash equilibrium in this setting. We emphasize that local Nash equilibria may fail to exist already in two-player zero-sum games (for example, $X_1=X_2=[0,1], u_1(x_1,x_2)=(x_1-x_2)^2, u_2=-u_1$). Notably, our adaptation of \stonr to general-sum games enables us to recover first-order Nash equilibria, representing a key contribution of our work.

\section{Nash Equilibria in Continuous Games}
Let $N= \{1,\dots,n\}$ be the set of players for some positive integer $n$. We assume that each player $i\in N$ has a compact convex set of strategies $X_i\subseteq \R^{d_i}$, where $d_i$ is a positive integer. We define 
\[
    \.X = \bigtimes_{i\in N}X_i
\]
and require that the utility function $u_i\colon \.X\to \R$ of each player $i\in N$ be continuous. The triple 
\[
    \G=(N,(X_i)_{i\in N},(u_i)_{i\in N})
\]
is then an $n$-player strategic game, which we call a \emph{continuous game}. The \emph{dimension} of game $\G$ is the integer 
\begin{equation}\label{def:dim}
    d = \sum_{i=1} ^n d_i.
\end{equation}
A vector $x_i=(x_{i1},\dots,x_{id_i})\in X_i$ denotes a strategy of player $i\in N$, and we define $$\.X_{-i}=\bigtimes_{\substack{j\in N\\j\neq i}} X_j.$$  A~\emph{Nash equilibrium (NE)} in continuous game $\G$ is a strategy profile $\.x^*=(x_1^*,\dots,x_n^*)\in\.X$ such that
\begin{equation}\label{def:NE}
    u_i(x_i,\.x_{-i}^*) \le u_i(\.x^*) \quad \text{for each $i\in N$ and all $x_i \in X_i$,}
\end{equation}
 where $\.x_{-i}^*=(x_1^*,\dots,x_{i-1}^*,x_{i+1}^*,\dots,x_n^*)\in\.X_{-i}$. 
 It is well known that if each function $u_i(.,\.x_{-i})$ is concave on~$X_i$ for every $\.x_{-i}\in \.X_{-i}$ in a continuous game $\G$, then $\G$ has a NE \cite{rosen1965existence}. Analogously to the Nash theorem for finite strategic games, every continuous game $\G$ has a NE in mixed strategies \cite{Glicksberg52}. However, the theory of mixed equilibria for continuous games is rather complicated, since mixed strategies may be arbitrary Borel probability measures with infinite supports \cite{Karlin59}. Several papers study computational methods for recovering mixed strategy NE in continuous games; cf. \cite{SteinOzdaglarParrilo08,LarakiLasserre12,Adam_Horcik_Kasl_Kroupa_2021}.

 There are many generalizations of the concept of NE. One such generalization is based on relaxing the global optimality requirement. A~\emph{local NE (LNE)} in continuous game $\G$ is a strategy profile $\.x^*=(x_1^*,\dots,x_n^*)\in\.X$ such that there exists $\delta>0$ satisfying for each $i\in N$ and every $x_i \in X_i$ the inequality
 \begin{equation}\label{def:LNE}
     u_i(x_i,\.x_{-i}^*) \le u_i(\.x^*) \quad \text{whenever $\norm{x_i-x_i^*}_2 < \delta$.}
 \end{equation}
 LNE is an important solution concept in min-max optimization, i.e., in case that $n=2$ and the game is zero-sum; see \cite{jin2020local} and references therein.

\section{First-Order Nash Equilibria}
From now on, we will assume that each utility function~$u_i$ in continuous game $\G$ is continuously differentiable and we call such game $\G$ \emph{continuously differentiable}.
Consider $[d_i]=\{1,\dots,d_i\}$ for each $i\in N$. For each $\.x=(x_1,\dots,x_n) \in \.X$, each $i\in N$, and each $j \in [d_i]$, we define
\begin{align}
    v_{ij}^{\G}(\.x) & = \frac{\partial u_i(\.x)}{\partial x_{ij}}, \label{def:pd}\\
    v_i^{\G}(\.x)& = \left( v_{i1}^{\G}(\.x),\dots,v_{id_{i}}^{\G}(\.x)\right),\; \text{and}\\ \label{def:v}
    v^{\G}(\.x) & = \left(v_1^{\G} (\.x),\dots,v_n^{\G}(\.x)\right).
\end{align}
By applying the first-order optimality conditions of nonlinear programming to each player's maximization problem based on \eqref{def:NE}, any LNE \(\mathbf{x}^*\) in a continuously differentiable game \(\mathcal{G}\) necessarily satisfies \(n\) variational inequalities \cite{facchinei2003finite} of the form 
\begin{equation}\label{VNE}
    \langle v_i^{\G}(\.x^*), x_i-x_i^* \rangle \le 0 \enskip \text{for each $i\in N$ and every $x_i\in X_i$.}
\end{equation}
If each $u_i(.,\.x_{-i})$ is concave on $X_i$ for every $\.x_{-i}\in \.X_{-i}$, then condition \eqref{VNE} implies that $\.x^*$ is a NE \cite{facchinei2003finite}. In general, condition \eqref{VNE} is necessary but not sufficient for~$\.x^*$ to be a LNE.  The concept of first-order Nash equilibrium is a useful weakening of the concept of LNE in the situations where LNE fails to exist or computing it is prohibitively expensive; cf. \cite{daskalakis2023stay, daskalakis2021non,Tsaknakis2021-jv,Nouiehed2019-ez,mertikopoulos2017convergence}.
\begin{definition}
    Let $\G$ be a continuously differentiable game. A~strategy profile $\.x^*$ satisfying \eqref{VNE} is called a \emph{first-order Nash equilibrium (FONE)}.
\end{definition}
FONE also appears in literature under different names, such as \emph{variational NE} \cite{rockafellar2018variational,rockafellar2024generalized} or \emph{game-stationary equilibrium} \cite{razaviyayn2020nonconvex}. We note that FONE is a special case of \emph{generalized NE} with nonconcave utilities \cite{Facchinei2010-su, Scutari2010-ly}. In case that $n=2$ and game $\G$ is zero-sum, $u_1+u_2=0$, any FONE is a \emph{min-max critical point} \cite{daskalakis2023stay}.

In the following, we omit the upper index ${\G}$ and write simply $v(\.x)$ in place of $v^{\G}(\.x)$ whenever game $\G$ is understood. By $\Pi_K(\.y)$ we denote the \emph{Euclidean projection} of a vector $\.y\in\R^k$ onto a closed convex set $K\subseteq \R^k$. Below we summarize the conditions equivalent to FONE appearing in literature, which can be deduced from the basic theory of variational inequalities \cite{facchinei2003finite}.

\begin{proposition}\label{VNEchar}
Let $\.x^* \in \.X$ be a strategy profile in a continuously differentiable game $\G$. The following are equivalent.
    \begin{enumerate}
        \item $\.x^*$ is a FONE.
        \item $\langle v(\.x^*),\.x-\.x^* \rangle \leq 0$ for all $\.x \in \.X$. \hfill \text{\normalfont \rmfamily VI$(v,\.X)$}
        \item $\.x^*=\Pi_{\.X}\left(\.x^*+v(\.x^*)\right)$.
    \end{enumerate}
\end{proposition}
Using item 3 in Proposition~\ref{VNEchar}, it is easy to see that every continuously differentiable game has at least one FONE. Define the mapping \(F \colon \.X \to \.X\) by \(F(\mathbf{x}) = \Pi_{\.X}\left(\mathbf{x} + v(\mathbf{x})\right)\), and observe that \(F\) is continuous. The existence of a FONE then follows directly from Brouwer's fixed-point theorem applied to the mapping \(F\). We refer to a point \(\mathbf{x}^*\) satisfying condition 2. in Proposition~\ref{VNEchar} as a solution to the variational inequality {\normalfont \rmfamily VI}\((v, \.X)\).

\section{Computing First-Order Nash Equilibria Using \stonr}
In this section, we make additional assumptions about strategy sets $X_i$ and mapping $v^{\G}$ defined by \eqref{def:v}. The first assumption will enable us to show that the variational inequality {\normalfont \rmfamily VI}$(v,\.X)$  can be decomposed into individual coordinates along the dimension $d$ of the game \eqref{def:dim}. This coordinatewise separation is in turn equivalent to the satisfaction of all coordinates, the property introduced by \citet[Definition 4]{daskalakis2023stay}. We recall the short notation $v_{ij}(\.x^*)$ for partial derivatives \eqref{def:pd}.
\begin{definition}\label{def:satisfaction}
Assume that $\G$ is a continuously differentiable game such that $X_i=[0,1]^{d_i}$ for each $i\in N$.  Let $\.x=(x_1,\dots,x_n) \in [0,1]^d$. For any $i \in N$ and $j \in [d_i]$, we call an ordered pair $ij$ a \emph{coordinate}, and we say that a~coordinate $ij$  is \emph{satisfied at} $\.x$ if one of these conditions holds:
    \begin{itemize}
        \item $v_{ij}(\.x) = 0.$ \hfill ($ij$ is \emph{zero-satisfied})
        \item $v_{ij}(\.x) < 0$ and $x_{ij} = 0.$ \hfill ($ij$ is \emph{boundary-satisfied})
        \item $v_{ij}(\.x) > 0$ and $x_{ij} = 1.$ \hfill ($ij$ is \emph{boundary-satisfied})
    \end{itemize}
\end{definition}
\begin{lemma}
    Let $\G$ be a continuously differentiable game and assume that $X_i=[0,1]^{d_i}$ for each $i\in N$. The following are equivalent for a point $\.x^* \in [0,1]^d$.
    \begin{enumerate}
        \item $\.x^*$ is a solution to {\normalfont \rmfamily VI$(v,[0,1]^d)$}.
        \item For each $i \in N$, every $j \in [d_i]$, and every $x_{ij} \in [0,1]$,
        \begin{equation}\label{VIij}
            v_{ij}(\.x^*)(x_{ij}-x_{ij}^*) \le 0.
        \end{equation}
        \item For each $i \in N$ and every $j \in [d_i]$, each coordinate $ij$ is satisfied at $\.x^*$.
    \end{enumerate}
\end{lemma}
\begin{proof}
    $1. \Rightarrow 2.$ Let $i \in N$, $j \in [d_i]$, and  $x_{ij} \in [0,1]$. Define $\.x \in \.X$ by 
    $$
    x_{kl} = \begin{cases}
        x_{ij} & \text{$k = i$ and $l = j$},\\
        x_{kl}^* & \text{otherwise,}
    \end{cases}
    $$
    for each $k\in [n]$ and $l \in [d_k]$.
     Then $$v_{ij}(\.x^*)(x_{ij} - x_{ij}^*) = \langle v(\.x^*),\.x - \.x^* \rangle \le 0.$$
     
     \noindent
    2. $\Rightarrow 1.$ This follows directly from the assumption by summing up the individual inequalities \eqref{VIij}, $$\langle v(\.x^*),\.x - \.x^*\rangle = \sum_{i=1}^n \sum_{j=1}^{d_i} v_{ij}(\.x^*)(x_{ij}-x_{ij}^*) \le 0.$$

\noindent
    2. $\Leftrightarrow 3.$ This follows immediately from the definitions.
\end{proof}
The assumption that each \(X_i\) is a unit hypercube \([0,1]^{d_i}\) is crucial for the design of the \stonr algorithm, which operates within the unit hypercube \([0,1]^d\). However, the algorithm can be extended to any convex compact joint strategy space \(\mathbf{X} = X_1 \times \dots \times X_n\), $X_i \subseteq \mathbb{R}^{d_i}$, if there exists a particular transformation \(H\) from \([0,1]^d\) to \(\mathbf{X}\) --- see \citet[Appendix B]{daskalakis2023stay} for details. Such a transformation \(H\) is straightforward to find if each \(X_i\) is a hyperrectangle,
$$
X_i = [a_{i1},b_{i1}] \times \dots \times [a_{id_i},b_{id_i}]
$$
for some \(a_{ij},b_{ij} \in \mathbb{R}\) with \(a_{ij} < b_{ij}\). Since many games have strategy spaces of this form, we will explain how to construct \(H\) in this case and describe the correspondence between the resulting variational inequalities. Let $$H = (H_{ij})_{i \in N, j \in [d_i]} \colon [0,1]^d \to \.X$$ be a bijective affine mapping with the component mappings $H_{ij} \colon [0,1]^d \to [a_{ij},b_{ij}]$ defined as
$
H_{ij}(\.x) = c_{ij}x_{ij}+a_{ij},
$
where $c_{ij} = b_{ij}-a_{ij}$.
Let $$w = (w_1,\dots,w_n) \colon [0,1]^d \to \mathbb{R}^n$$ where $w_i \colon [0,1]^{d} \to \mathbb{R}$ are defined as $w_i = u_i \circ H$.
Consider the map $\hat{w} = (\hat{w}_{ij})_{i \in N, j \in [d_i]} \colon [0,1]^d \to \mathbb{R}^d$ given by
$$
\hat{w}_{ij}(\.x) = \frac{\partial w_i(\.x)}{\partial x_{ij}}.
$$
For every $i \in N$, $j \in [d_i]$ and $\.x \in [0,1]^d$, the chain rule yields
\begin{equation}\label{eq:w}
    \hat{w}_{ij}(\.x) = c_{ij}v_{ij}(H(\.x)).
\end{equation}
\begin{proposition}\label{pro:transform}
     Let $\.x^* \in [0,1]^d$. The following are equivalent.
    \begin{enumerate}
        \item $\.x^*$ is a solution to {\normalfont \rmfamily VI}$(\hat{w},[0,1]^d)$.
        \item $H(\.x^*)$ is a solution to {\normalfont \rmfamily VI}$(v,\.X)$.
    \end{enumerate}
\end{proposition}
\begin{proof}
        $1.\Rightarrow 2.$ We must show $\langle v(H(\.x^*)),\.x - H(\.x^*)\rangle \leq 0$ for all $\.x \in \.X$. Using \eqref{eq:w} the left-hand side equals
        \begin{align*}
        &\langle v(H(\.x^*)),\.x - H(\.x^*)\rangle \\& = \sum_{i=1}^n \sum_{j=1}^{d_i} v_{ij}(H(\.x^*))(x_{ij}-H_{ij}(\.x^*)) \\
        &=\sum_{i=1}^n \sum_{j=1}^{d_i} \frac{\hat{w}_{ij}(\.x^*)}{c_{ij}}(x_{ij}-H_{ij}(\.x^*))\\
        &=\sum_{i=1}^n \sum_{j=1}^{d_i} \hat{w}_{ij}(\.x^*)\left(\frac{x_{ij}-a_{ij}}{c_{ij}}-x^*_{ij}\right).
        \end{align*}
        Let $\.y \in \mathbb{R}^d$ be a vector with coordinates $y_{ij} = \frac{x_{ij}-a_{ij}}{c_{ij}}$, where necessarily $0 \leq y_{ij} \leq 1$. Then the last term above equals $\langle \hat{w}(\.x^*),\.y-\.x^* \rangle$ and it is nonnegative by assumption~1. Implication $2.\Rightarrow 1.$ is proved similarly.
\end{proof}
Proposition~\ref{pro:transform} demonstrates a method for solving {\normalfont \rmfamily VI}$(v,\.X)$ by reducing the problem to solving {\normalfont \rmfamily VI}$(\hat{w},[0,1]^d)$. Specifically, any solution $\.x^* \in [0,1]^d$ to {\normalfont \rmfamily VI}$(\hat{w},[0,1]^d)$ can be mapped by $H$ to a corresponding solution $H(\.x^*) \in \.X$ for {\normalfont \rmfamily VI}$(v,\.X)$.

\subsection{Sketch of \stonr Algorithm for n-Player Games}
In the remainder of this paper, we will make two additional assumptions about the properties of the mapping \(v = v^{\G}\) associated with the game \(\G\). These assumptions are crucial for ensuring the convergence of the \stonr algorithm.
\begin{itemize}
    \item Mapping $v$ is $L$-Lipschitz,
    $$
    \norm{v(\.x)-v(\.y)}_2 \le L \cdot \norm{\.x-\.y}_2,
    $$
    for some $L>0$ and every $\.x,\.y \in \.X$.
    \item The Jacobian $\nabla v$ of $v$ is $\Lambda$-Lipschitz,
     $$
    \norm{\nabla v(\.x)-\nabla v(\.y)}_2 \le \Lambda \cdot \norm{\.x-\.y}_2,
    $$
    for some $\Lambda>0$ and every $\.x,\.y \in \.X$.  
\end{itemize}
The last property is called $\Lambda$-smoothness in \cite{daskalakis2023stay}.
Below we present a sketch of the \stonr algorithm for $n$-player nonzero-sum games. It is important to note that we are not introducing a  new algorithm but rather showing that the existing one \cite{daskalakis2023stay} can be applied to a broader class of games. The only required modification is in the algorithm's input $v$, which now corresponds to the $n$-player nonzero-sum game $\mathcal{G}$. To make our arguments clearer, we will first present the continuous-time version of the algorithm, followed by a brief overview of the discrete-time version.
\subsection{Continuous-time Version of \stonr}
The goal is to find a~point $\.x^* \in [0,1]^d$ such that each coordinate is satisfied at $\.x^*$ according to Definition~\ref{def:satisfaction}. The algorithm is initialized at $\.x(0) = (0,\dots,0)$ and aims to satisfy all coordinates one-by-one.
% We will now label the coordinates with an index $i \in [d]$.
For $i \in N$ and $j \in [d_i]$, define the set of coordinates $$\Smaller(ij) = \{kl \mid k \in N, l \in [d_k], kl < ij\},$$
where $<$ refers to the lexicographical ordering.
The algorithm proceeds in epochs, each denoted by a pair $(ij,S)$, where $S \subseteq \Smaller(ij)$ is the set of already zero-satisfied coordinates.
When the algorithm starts an epoch $(ij,S)$ at time $t$ and at some point $\.x \in [0,1]^d$, it is assumed that all coordinates in $S$ are zero-satisfied at~$\.x$, all coordinates in $\Smaller(ij) \setminus S$ are boundary-satisfied at $\.x$, and the epoch's goal is to find a point $\mathbf{x}' \in [0,1]^d$ such that all coordinates $\leq ij$ are satisfied at $\.x'$. Let us define the direction of movement in each epoch to achieve this goal.

For $i \in N$, $j \in [d_i]$, any set of coordinates $S = \{s_1,\dots,s_m\} \subseteq \Smaller(ij)$, and $\mathbf{x} \in [0,1]^d$,
define $\mathbf{u} \in \mathbb{R}^d$ such that
\begin{enumerate}
    \item \(u_{kl} = 0\), for all \(kl \notin S \cup \{ij\}\),
    \item \(\langle \nabla v_{kl}(\mathbf{x}), \mathbf{u} \rangle = 0\), for all \(kl \in S\), and
    \item the sign of determinant
            $$
            \left|
            \begin{matrix}
            \dfrac{\partial v_{s_1}(\mathbf{x})}{\partial x_{s_1}} & \dfrac{\partial v_{s_2}(\mathbf{x})}{\partial x_{s_1}} & \cdots & \dfrac{\partial v_{s_m}(\mathbf{x})}{\partial x_{s_1}} & u_{s_1} \\\\
            \vdots & \vdots & \vdots & \vdots & \vdots\\\\
            \dfrac{\partial v_{s_1}(\mathbf{x})}{\partial x_{s_m}} & \dfrac{\partial v_{s_2}(\mathbf{x})}{\partial x_{s_m}} & \cdots & \dfrac{\partial v_{s_m}(\mathbf{x})}{\partial x_{s_m}} & u_{s_m} \\\\
            \dfrac{\partial v_{s_1}(\mathbf{x})}{\partial x_{ij}} & \dfrac{\partial v_{s_2}(\mathbf{x})}{\partial x_{ij}} & \cdots & \dfrac{\partial v_{s_m}(\mathbf{x})}{\partial x_{ij}} & u_{ij}
            \end{matrix}
            \right|
            $$
    equals the sign of \((-1)^{|S|}\).
\end{enumerate}

The first two constraints reduce the problem to identify only non-zero coordinates of $\mathbf{u}$, determined by solving a~system of $m$ linear equations with $m+1$ variables. If the associated kernel is 1-dimensional, the solution space forms a line, and the third constraint specifies the direction of this line. If there exists a unique unit vector that satisfies all three constraints, we denote it by $D_{S}^{ij}(\mathbf{x})$.

The algorithm tries to achieve the epoch's goal by executing the continuous-time dynamics $\{\.z(\tau)\}_{\tau \geq 0}$ initialized at $\.z(0) = \.x(t)$ and moving in the direction $D^{ij}_S(\.z(\tau))$ within $[0,1]^d$. The continuous-time dynamics of epoch $(ij,S)$ ends at one of the \emph{exit points} $\.x' \in [0,1]^d$. They are categorized as follows:
\begin{itemize}
    \item \emph{Good exit point}: Coordinate $ij$ is satisfied at $\.x'$.
    \item \emph{Bad exit point}: For some coordinate $kl \in S \cup \{ij\}$, moving in the direction $D^{ij}_S(\.x')$ will eventually cause the $kl$-th coordinate to exceed the bounds of $[0,1]$.
    \item \emph{Middling exit point}: For some boundary-satisfied coordinate $kl \in \Smaller(ij) \setminus S$,
    continuing in the direction $D^{ij}_S(\.x')$ will cause $kl$ to become unsatisfied.
\end{itemize}
\begin{algorithm}[tb]
\caption{Continuous-time \stonr for $n$-player games}
\label{alg:algorithm}
\textbf{Input:} Mapping \( v \)\\
\textbf{Output:} Solution to {\normalfont \rmfamily VI}\((v, [0,1]^d)\)
\begin{algorithmic}[1] %[1] enables line numbers
    \STATE Initially $\.x(0) \leftarrow (0, \ldots, 0)$, $ij \leftarrow 11$, $S \leftarrow \emptyset$, $t \leftarrow 0$.
    \WHILE{$\.x(t)$ is not a solution to {\normalfont \rmfamily VI}$(v,[0,1]^d)$}
        \STATE Initialize the continuous-time dynamics of epoch $(ij,S)$, $\dot{\.z}(\tau) = D^{ij}_S(\.z(\tau))$, at $\.z(0) = \.x(t)$.
        % \STATE Initialize epoch $(ij,S)$'s continuous-time dynamics, $\dot{\.z}(\tau) = D^{ij}_S(\.z(\tau))$, at $\.z(0) = \.x(t)$.
        \WHILE{$\.z(\tau)$ is not an exit point}
            \STATE Execute $\dot{\.z}(\tau) = D^{ij}_S(\.z(\tau))$ forward in time.
        \ENDWHILE
        \STATE Set $\.x(t+\tau_{exit}) = \.z(\tau_{exit})$ (where $\tau_{exit}$ is the time $\.z(\tau)$ became an exit point).
        \IF{$\mathbf{x}(t+\tau_{exit})$ is a good exit point}
            \IF{$ij$ is zero-satisfied at $\.x(t+\tau_{exit})$}
                \STATE Add $ij$ to $S$.
                % \STATE Update $S \leftarrow S \cup \{ij\}$.
            \ENDIF
            \STATE Increase $ij$ by one in the lexicographical order.
            % \STATE Update $i \leftarrow i + 1$.
        \ELSIF{$\mathbf{x}(t+\tau_{exit})$ is a bad exit point for $kl = ij$}
            \STATE Decrease $ij$ by one in the lexicographical order to get $ij_{new}$, and remove $ij_{new}$ from $S$.
            % \STATE Update $i \leftarrow i - 1$ and $S \leftarrow S \setminus \{i-1\}$.
        \ELSIF{$\mathbf{x}(t+\tau_{exit})$ is a bad exit point for $kl \neq ij$}
        \STATE Remove $kl$ from $S$.
            % \STATE Update $S \leftarrow S \setminus \{j\}$.
        \ELSIF{$\mathbf{x}(t+\tau_{exit})$ is a middling exit point for $kl < ij$}
        \STATE Add $kl$ to $S$.
            % \STATE Update $S \leftarrow S \cup \{j\}$.
        \ENDIF
        \STATE Set $t \leftarrow t + \tau_{exit}$.
    \ENDWHILE
    \STATE \textbf{return} $\.x(t)$
\end{algorithmic}
\end{algorithm}
Algorithm~\ref{alg:algorithm} provides the pseudocode for the continuous-time version of \stonr. The convergence of Algorithm~\ref{alg:algorithm} to a FONE in two-player zero-sum games is proved under mild assumptions --- see \citet[Assumptions 1, 2, 3]{daskalakis2023stay}. These assumptions are as follows:
\begin{enumerate}
    \item For all $\.x \in [0,1]^d$, each $i \in N$, $j \in [d_i]$, and every $S \subseteq \Smaller(ij)$, if all coordinates $kl \in S$ are zero-satisfied at $\.x$ and for all coordinates $kl \notin S \cup \{ij\}$ it holds $\.x_{kl} \in \{0,1\}$, then the direction $D^{ij}_S(\.x)$ is uniquely defined.
    \item At any time, only one coordinate can trigger a middling or a bad event.
    \item It is possible to determine whether a coordinate begins or stops being satisfied by looking at the Jacobian of $v$.
\end{enumerate}

Assumptions 1 and 3 are verified during the direction computation, while Assumption 2 is checked at each visited point. In our code, we specifically verify Assumption 1, as it is crucial for ensuring a unique direction. The original paper \cite{daskalakis2023stay} emphasizes that all three assumptions are mild, a point which remains valid in the context of $n$-player nonzero-sum games. We emphasize that 
assumptions 1, 2, and 3 pertain solely to the analytical properties of the utility functions and do not depend on the zero-sum property.

We argue that Algorithm~\ref{alg:algorithm} converges to a FONE even when $v$ is associated with an $n$-player nonzero-sum game $\mathcal{G}$. This is based on the following observations.
\begin{enumerate}
    \item Min-max critical points in game $\G$ coincide with FONE whenever $\G$ is a two-player zero-sum continuously differentiable game.
    \item \stonr algorithm finds a solution to the non-monotone variational inequality associated with the Nash equilibrium problem for $\G$.
    \item The assumptions 1, 2, and 3 crucial for the \stonr convergence do not rely on the two-player or the zero-sum property of $\mathcal{G}$.
    \item The convergence analysis of the \stonr algorithm involves only the properties of the map $v$, and does not use the zero-sum property of the underlying game $\G$.
    \item The behavior of \stonr does not depend on the number of players since the concept of coordinatewise satisfaction uses only the vector of ``concatenated partial gradients'' $v(\.x)$.
\end{enumerate}
\subsection{Discrete-time Version of \stonr}
Algorithm~\ref{alg:algorithm} can be adapted to a discrete-time version \citet[Dynamics 4]{daskalakis2023stay}, which closely follows the structure of the continuous-time version. The main adjustment is in step 5, where the update rule is changed to $\.z^{(k+1)} \leftarrow \.z^{(k)} + \gamma \cdot D^{ij}_S(\.z^{(k)})$, with $\gamma > 0$ representing a step size. This version operates with $(\epsilon,\gamma)$-exit points, where $\epsilon > 0$ is the exit point error. 
It was shown \cite{daskalakis2023stay} that the discrete-time version of Algorithm~\ref{alg:algorithm} converges to an approximate solution to {\normalfont \rmfamily VI}$(v,[0,1]^d)$ in two-player zero-sum games. An \emph{$\alpha$-approximate solution} to {\normalfont \rmfamily VI}$(v,[0,1]^d)$, $\alpha > 0$, is a point $\.x^* \in [0,1]^d$ satisfying
\(
    \langle v(\.x^*), \.x - \.x^* \rangle \leq \alpha \) for all $\.x \in [0,1]^d$.

The original convergence theorem is presented below.
\begin{theorem}[{\citet[Theorem 39]{daskalakis2023stay}}]\label{thm:discrete}
Let assumptions 1, 2, and 3 hold. For every $\alpha > 0$ there exist constants $\epsilon$, $\gamma$, $\bar{M}$, $K$  such that the discrete-time version of Algorithm~\ref{alg:algorithm} with step size $\gamma$ and error $\epsilon$ finishes after $M \leq \bar{M}$ iterations of the while loop at line 2, and $\.x^{(M)}$ is an $\alpha$-approximate solution to {\normalfont \rmfamily VI}$(v,[0,1]^d)$. Moreover, for each iteration \( m \leq M \) of the while loop at line 2, the while loop at line 4 executes at most \( K \) iterations.
\end{theorem}
Theorem~\ref{thm:discrete} directly implies the convergence of the discrete-time version of Algorithm~\ref{alg:algorithm} to an approximate FONE in $n$-player nonzero-sum games. We present the proof below.
\begin{proof}
Our claim is an immediate consequence of Theorem~\ref{thm:discrete} because all the assumptions of the theorem are satisfied also for $n$-player nonzero-sum games. The assumptions are as follows:
\begin{enumerate}
    \item Algorithm~\ref{alg:algorithm} converges to an exact solution to {\normalfont \rmfamily VI}$(v,[0,1]^d)$. Our discussion of the continuous-time version fulfills this.
    \item The mapping $v$ is Lipschitz, which is true by the assumption.
        \item The mapping \( D^{ij}_S \) is Lipschitz, as directly established by Lemma~25 in \cite{daskalakis2023stay}. Notably, the lemma applies without requiring the game \( \mathcal{G} \) to be two-player or zero-sum.
\end{enumerate}
In conclusion, the discrete-time version of Algorithm~\ref{alg:algorithm} converges to an approximate FONE in our setting.
\end{proof}

\section*{Numerical Experiments}
We demonstrate the versatility of the \stonr extension by solving a variety of games from recent research papers. The transformation outlined in Proposition~\ref{pro:transform} was applied to hyperrectangle strategy sets. All experiments were executed on a laptop with Intel Core i5 CPU and 8 GB RAM. We implemented the discrete-time version of the \stonr algorithm in Julia, utilizing the \texttt{Symbolics.jl} package for differentiation and the \texttt{LinearAlgebra.jl} package for matrix operations. The source code will be included as an attachment to this paper.

\begin{example}\label{ex:example1}
(Two-player nonzero-sum nonconcave game from \cite{SteinOzdaglarParrilo08})
    Let $n = 2$. Each strategy set $X_i$ is equal to $[-1,1]$. The utility functions are
    \begin{align*}
        u_1(x_{11},x_{21}) &= 2x_{11}x_{21}+3x_{21}^3-2x_{11}^3-x_{11}-3x_{11}^2x_{21}^2,\\
        u_2(x_{11},x_{21}) &= 2x_{11}^2x_{21}^2-4x_{21}^3-x_{11}^2+4x_{21}+x_{11}^2x_{21}.
    \end{align*}
    \stonr finds a FONE $\.x^* = (-1,-1)$, which is a global maximum in $x_{11}$ and a local maximum in $x_{21}$. Therefore, the point $\.x^*$ is a local NE. The mapping $v$ has a value of $v(\.x^*) = (-3,-11)$. The solve time was $2 \, \text{ms}$.
\end{example}

\begin{example}\label{ex:example2}
(Three-player nonzero-sum nonconcave game from \cite{SteinOzdaglarParrilo08})
    Let $n = 3$. Each strategy set $X_i$ is equal to $[-1,1]$. The utility functions are
    \begin{align*}
    u_1(x_{11},x_{21},x_{31}) &= 1+2x_{11}+3x_{11}^2+2x_{21}x_{31}\\&+4x_{11}x_{21}x_{31}+6x_{11}^2x_{21}x_{31}+3x_{21}^2x_{31}^2\\&+6x_{11}x_{21}^2x_{31}^2+9x_{11}^2x_{21}^2x_{31}^2,\\
        u_2(x_{11},x_{21},x_{31}) &= 7+2x_{11}+3x_{11}^2+2x_{21}+4x_{11}x_{21}\\&+6x_{11}^2x_{21}+3x_{31}^2+6x_{11}x_{31}^2+9x_{11}^2x_{31}^2,\\
        u_3(x_{11},x_{21},x_{31}) &= -x_{31}-2x_{11}x_{31}-3x_{11}^2x_{31}\\&-2x_{21}x_{31}-4x_{11}x_{21}x_{31}-6x_{11}^2x_{21}x_{31}\\&-3x_{21}x_{31}^2-6x_{11}x_{21}x_{31}^2-9x_{11}^2x_{21}x_{31}^2.
    \end{align*}
    In this game, the \stonr method follows the trajectory depicted in Figure~\ref{fig:example2} and finds the FONE $\.x^* = (-1,1,-0.5)$. The FONE $\.x^*$ is a local maximum in $x_{11}$, and a global maximum in $x_{21}$ and $x_{31}$. Therefore, the point $\.x^*$ is a local NE.
    The mapping $v$ has a value of $v(\.x^*) = (-3,4,0)$. \stonr was executed with parameters $\gamma = 10^{-3}$ and $\epsilon = 10^{-2}$ and the solve time was $7 \,\text{ms}$.
    \begin{figure}
        \footnotesize
        \centering
        \input{ex2trajectoryt}
        \caption{The trajectory of \stonr in Example~\ref{ex:example2}  starting at $(-1,-1,-1)$ and converging to the FONE depicted by the cross.}
        \label{fig:example2}
    \end{figure}
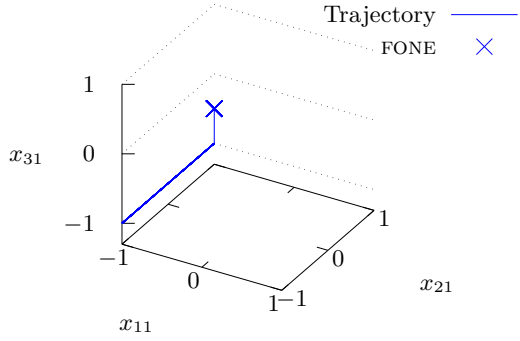
\end{example}

\begin{example}\label{ex:example3}
(Six-player nonzero-sum concave game, power control problem in optical networks \cite{pan2007global})
The optical signal-to-noise ratio (OSNR) optimization problem is formulated as an $n$-player strategic game. Each player selects their strategy from a one-dimensional interval $X_i = [x_{\min},x_{\max}]$. The strategy $x_{i1}$ of the $i$-th player represents the $i$-th channel signal power.
The cost function for each player $i$ is given by:
$$
J_i(\.x) = \eta_ix_{i1}-\beta_i \left(\ln \left(1+a_i\frac{\gamma_i(\.x)}{1-\Phi_{ii}\gamma_i(\.x)}    \right)-x_{i1}\right),
$$
where $\eta_i$, $\beta_i$, and $a_i$ are channel-specific parameters. The parameter $\gamma_i(\.x)$ represents the OSNR, defined as
$$
\gamma_i(\.x) = \frac{x_{i1}}{n_0+\sum_j\Phi_{ij}x_{j1}},
$$
for a given system matrix $\Phi$ and input noise $n_0$. The numerical setting is adopted from~\cite{nguyen2023nash}. Specifically, $n=6$, $\eta_1=\dots=\eta_6=1$,
\begin{align*}
    \beta & = [0.5, 0.51, 0.52, 0.3, 0.31, 0.32],\\ a & = [0.261, 0.494, 0.107, 0.366, 0.208, 0.305],
\end{align*}
 $x_{\min} = 0.2 \, \text{mW}$, $x_{\max} = 2 \, \text{mW}$, and the matrix 
\[ \Phi\! =\tfrac{1}{10^5}\!\!\! \begin{bmatrix}

7.463 & 7.378 & 7.293 & 7.210 & 7.127 & 6.965 \\
7.451 & 7.365 & 7.281 & 7.198 & 7.115 & 6.953 \\
7.438 & 7.353 & 7.269 & 7.186 & 7.103 & 6.942 \\
7.427 & 7.342 & 7.258 & 7.175 & 7.093 & 6.931 \\
7.409 & 7.324 & 7.240 & 7.157 & 7.075 & 6.914 \\
7.387 & 7.303 & 7.219 & 7.136 & 7.055 & 6.894
\end{bmatrix}
.
\]
 \stonr finds a FONE of this game at $\.x^* = (0.333, 0.337, 0.340, 0.230, 0.236, 0.241)$, which also represents a NE.
 The mapping $v$ has a value of $v(\.x^*) = (0,0,0,0,0,0)$.
 \stonr was executed with parameters $\gamma = 10^{-3}$ and $\epsilon = 10^{-2}$ and the solve time was $10 \,\text{ms}$.
\end{example}

\begin{example}\label{ex:example4}
    (Two-player zero-sum game \cite{parrilo2006polynomial})
    Let $n = 2$. Each strategy set $X_i$ is equal to $[-1,1]$. The utility functions are
    \begin{align*}
        u_1(x_{11},x_{21}) &= 2x_{11}x_{21}^2-x_{11}^2-x_{21},\\
        u_2(x_{11},x_{21}) &= -u_1(x_{11},x_{21}).
    \end{align*}
    The \stonr method finds a FONE $\.x^* = (0.394,0.632)$, which is the only NE of the game. The mapping $v$ has a value of $v(\.x^*) = (0,0)$.
    \stonr was executed with parameters $\gamma = 10^{-3}$ and $\epsilon = 10^{-2}$ and the solve time was $13 \,\text{ms}$. The trajectory is illustrated in~Figure~\ref{fig:example4}.
    \begin{figure}
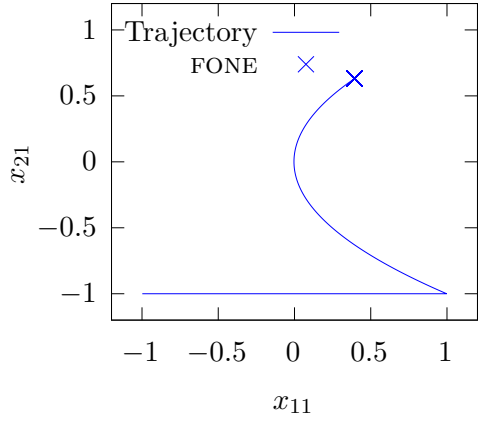

    \centering
    \include{ex4trajectoryt}
    \caption{The trajectory of \stonr in Example~\ref{ex:example4} starting at $(-1,-1)$ and converging to the FONE depicted by a cross.}
    \label{fig:example4}
    \end{figure}
\end{example}

\begin{example}\label{ex:example5}
(Three-player nonzero-sum game from \cite{kroupa2021separable})
    Let $n = 3$. Each strategy set $X_i$ is equal to $[0,1]$. The utility functions are
    \begin{align*}
        u_1(x_{11},x_{21},x_{31}) &= -2x_{11}x_{21}^2-2x_{11}^2+5x_{11}x_{21}\\&-4x_{11}x_{31}-x_{21}-2x_{31},\\
        u_2(x_{11},x_{21},x_{31}) &= 2x_{11}x_{21}^2-2x_{21}x_{31}^2-2x_{11}^2-5x_{11}x_{21}\\&-2x_{21}^2+5x_{21}x_{31}+x_{21},\\
        u_3(x_{11},x_{21},x_{31}) &= 2x_{21}x_{31}^2+4x_{11}^2+4x_{11}x_{31}+2x_{21}^2\\&-5x_{21}x_{31}+2x_{31}.
    \end{align*}
    The \stonr method finds a FONE of this game at $\.x^* = (0,1,1)$ and follows the trajectory depicted in Figure~\ref{fig:example5}. It is a global maximum in $x_{11}$ and $x_{21}$, and a local maximum in $x_{31}$. Therefore, the point $\.x^*$ is a local NE. The mapping $v$ has a value of $v(\.x^*) = (-1,0,1)$. \stonr was executed with parameters $\gamma = 10^{-3}$ and $\epsilon = 10^{-2}$ and the solve time was $21 \,\text{ms}$.
    \begin{figure}
        \footnotesize
        \centering
        \input{ex5trajectoryt}
        \caption{\stonr trajectory of Example~\ref{ex:example5}  starting at $(0,0,0)$ and converging to the FONE depicted by a cross.}
        \label{fig:example5}
    \end{figure}
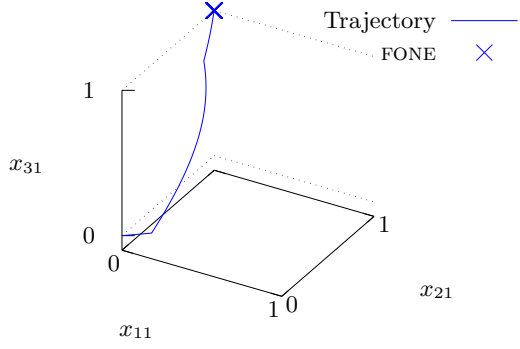
\end{example}

\begin{example}[Nonzero-sum Adversarial Hypothesis Testing Game \cite{Yasodharan}]\label{ex:example6}
    This example is a two-player nonzero-sum game played between an attacker and a defender over $m$ data points. The defender picks actions $\boldsymbol{\phi}$ from the set $[0,1]^m$, while the attacker picks $q$ from the interval $[0, 1]$. The utility function of the defender parametrized~by a positive $\xi$ is
    \[       
        u_1(\boldsymbol{\phi},q) = \sum_{i=0}^{m}\phi_i  \binom{m}{i}  q^i(1-q)^{m-i} - \frac{\xi}{2^m}\sum_{i=0}^{m}\phi_i  \binom{m}{i} - 1
    \]
    and the utility function of the attacker is
        \begin{align*}        
        u_2(\boldsymbol{\phi},q) &= - c(q) + \sum_{i=0}^{m}\binom{m}{i}(1-\phi_i)q^i(1-q)^{m-i},
    \end{align*}
    where $c(q) = (q-0.8)^2$. The \stonr method converges to NE when $m \le 2$ and to FONE otherwise. For example, instantiated with $m=3$ and $\xi=0.2$, \stonr converges to the LNE $\approx ([0.28, 1, 1, 1], 0.7)$. While the profile is optimal for the defender, the attacker could achieve better utility by playing $0$. A local NE is similarly recovered for $\xi=0.8$; whereas for $\xi=1.2$ and $3.2$, \stonr converges to NE.
    
    Trajectories for $m=1$ and for different values of $\xi$ are shown in Figure~\ref{fig:example6}. The \stonr execution parameters were $\gamma = 10^{-4}$ and $\epsilon = 10^{-4}$. The average solve time for $m=1$ was $93 \,\text{ms}$. 
    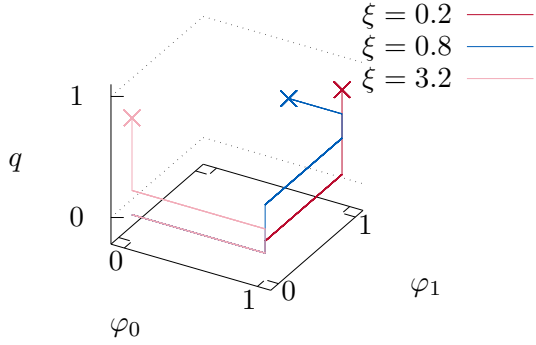
\begin{figure}
        \centering
        \input{hypotrajectoryt}
        \caption{\stonr trajectories for the game in Example~\ref{ex:example6} with $m=1$ and $\xi \in (0.2, 0.8, 3.2)$. The trajectories start at $(0,0,0)$ and converge to FONEs depicted by crosses.}
        \label{fig:example6}
    \end{figure}
\end{example}

\begin{example}[Two-player zero-sum nonconcave game from \cite{karlin2003mathematical}]
    The strategy sets for both players are unit intervals $X_1 = X_2 = [0,1]$. The utility function for the first player is defined by
\begin{align*}
    u_1(x_{11}, x_{21}) = 
    &\left( x_{21} - \frac{1}{2} \right) 
    \left[
    \frac{1 + \left( x_{11} - \frac{1}{2} \right) \left( x_{21} - \frac{1}{2} \right)^2}
         {1 + \left( x_{11} - \frac{1}{2} \right)^2 \left( x_{21} - \frac{1}{2} \right)^4} \right. \\
    &\left. - \frac{1}{1 + \left( \frac{x_{11}}{3} - \frac{1}{2} \right)^2 \left( x_{21} - \frac{1}{2} \right)^4} 
    \right].
\end{align*}

    \stonr finds a FONE $\.x^* = (0,0.497)$, which is not a local NE. The mapping $v$ has a value of $v(\.x^*) = (0,0)$. \stonr was executed with parameters $\gamma = 10^{-5}$ and $\epsilon=10^{-5}$, and the solve time was $84 \,\text{ms}$.
    We note that the unique equilibrium strategy for each player is the Cantor distribution. However, the FONE lies in rational numbers and is not hard to find.
\end{example}

In our numerical examples, we note that in many cases, the FONE found aligns with a LNE, preserving the local optimality of equilibria. Verifying that a FONE is indeed a LNE is complex, requiring the second-order sufficiency conditions for local optima \cite[p. 252]{bertsekas2016nonlinear}. In general, analyzing when FONE coincides with LNE as solution concepts across families of games is nontrivial; see Proposition 7 by \citet{pang2011nonconvex} for further details.

\section*{Conclusions}

While LNE can be insightful in modeling real-world games by capturing local deviations, FONE offers a more tractable solution concept for games in which strategy sets are compact and convex. Given that even LNE may fail to exist in nonconcave games, focusing on first-order solutions is a natural relaxation. FONE may not always qualify as LNE, as equation (7) does not generally imply LNE. However, in numerous examples, we demonstrated that the FONE identified is indeed an LNE. A thorough analysis of this phenomenon would require verifying second-order sufficiency conditions, which can be complex and is beyond the scope of this paper. It is important to note that the algorithm is not designed to find second-order solutions. However, it has guaranteed convergence to a FONE in games without any assumptions on the utility functions, except for their smoothness. 

One potential use for \stonr could be as a preliminary step for methods that identify stronger solution concepts, such as local or approximate NE. To check if FONE is a NE, it is only necessary to check if each player is playing their best response. 

Since hyperrectangle strategy sets may be restrictive in some applications, another promising future research direction is to extend \stonr to more general convex strategy sets; see \citet[Appendix B]{daskalakis2023stay}.

\bibliography{bibliography}
\bibliographystyle{plainnat}

\end{document}

%% file: ex2trajectoryt.tex
% GNUPLOT: LaTeX picture with Postscript
\begingroup
  \makeatletter
  \providecommand\color[2][]{%
    \GenericError{(gnuplot) \space\space\space\@spaces}{%
      Package color not loaded in conjunction with
      terminal option `colourtext'%
    }{See the gnuplot documentation for explanation.%
    }{Either use 'blacktext' in gnuplot or load the package
      color.sty in LaTeX.}%
    \renewcommand\color[2][]{}%
  }%
  \providecommand\includegraphics[2][]{%
    \GenericError{(gnuplot) \space\space\space\@spaces}{%
      Package graphicx or graphics not loaded%
    }{See the gnuplot documentation for explanation.%
    }{The gnuplot epslatex terminal needs graphicx.sty or graphics.sty.}%
    \renewcommand\includegraphics[2][]{}%
  }%
  \providecommand\rotatebox[2]{#2}%
  \@ifundefined{ifGPcolor}{%
    \newif\ifGPcolor
    \GPcolortrue
  }{}%
  \@ifundefined{ifGPblacktext}{%
    \newif\ifGPblacktext
    \GPblacktexttrue
  }{}%
  % define a \g@addto@macro without @ in the name:
  \let\gplgaddtomacro\g@addto@macro
  % define empty templates for all commands taking text:
  \gdef\gplbacktext{}%
  \gdef\gplfronttext{}%
  \makeatother
  \ifGPblacktext
    % no textcolor at all
    \def\colorrgb#1{}%
    \def\colorgray#1{}%
  \else
    % gray or color?
    \ifGPcolor
      \def\colorrgb#1{\color[rgb]{#1}}%
      \def\colorgray#1{\color[gray]{#1}}%
      \expandafter\def\csname LTw\endcsname{\color{white}}%
      \expandafter\def\csname LTb\endcsname{\color{black}}%
      \expandafter\def\csname LTa\endcsname{\color{black}}%
      \expandafter\def\csname LT0\endcsname{\color[rgb]{1,0,0}}%
      \expandafter\def\csname LT1\endcsname{\color[rgb]{0,1,0}}%
      \expandafter\def\csname LT2\endcsname{\color[rgb]{0,0,1}}%
      \expandafter\def\csname LT3\endcsname{\color[rgb]{1,0,1}}%
      \expandafter\def\csname LT4\endcsname{\color[rgb]{0,1,1}}%
      \expandafter\def\csname LT5\endcsname{\color[rgb]{1,1,0}}%
      \expandafter\def\csname LT6\endcsname{\color[rgb]{0,0,0}}%
      \expandafter\def\csname LT7\endcsname{\color[rgb]{1,0.3,0}}%
      \expandafter\def\csname LT8\endcsname{\color[rgb]{0.5,0.5,0.5}}%
    \else
      % gray
      \def\colorrgb#1{\color{black}}%
      \def\colorgray#1{\color[gray]{#1}}%
      \expandafter\def\csname LTw\endcsname{\color{white}}%
      \expandafter\def\csname LTb\endcsname{\color{black}}%
      \expandafter\def\csname LTa\endcsname{\color{black}}%
      \expandafter\def\csname LT0\endcsname{\color{black}}%
      \expandafter\def\csname LT1\endcsname{\color{black}}%
      \expandafter\def\csname LT2\endcsname{\color{black}}%
      \expandafter\def\csname LT3\endcsname{\color{black}}%
      \expandafter\def\csname LT4\endcsname{\color{black}}%
      \expandafter\def\csname LT5\endcsname{\color{black}}%
      \expandafter\def\csname LT6\endcsname{\color{black}}%
      \expandafter\def\csname LT7\endcsname{\color{black}}%
      \expandafter\def\csname LT8\endcsname{\color{black}}%
    \fi
  \fi
    \setlength{\unitlength}{0.0500bp}%
    \ifx\gptboxheight\undefined%
      \newlength{\gptboxheight}%
      \newlength{\gptboxwidth}%
      \newsavebox{\gptboxtext}%
    \fi%
    \setlength{\fboxrule}{0.5pt}%
    \setlength{\fboxsep}{1pt}%
    \definecolor{tbcol}{rgb}{1,1,1}%
\begin{picture}(4240.00,3400.00)%
    \gplgaddtomacro\gplbacktext{%
      \csname LTb\endcsname%%
      \put(1114,980){\makebox(0,0){\strut{}$-1$}}%
      \csname LTb\endcsname%%
      \put(1713,807){\makebox(0,0){\strut{}$0$}}%
      \csname LTb\endcsname%%
      \put(2312,634){\makebox(0,0){\strut{}$1$}}%
      \csname LTb\endcsname%%
      \put(2450,678){\makebox(0,0){\strut{}$-1$}}%
      \csname LTb\endcsname%%
      \put(2796,978){\makebox(0,0){\strut{}$0$}}%
      \csname LTb\endcsname%%
      \put(3142,1277){\makebox(0,0){\strut{}$1$}}%
      \csname LTb\endcsname%%
      \put(973,1240){\makebox(0,0)[r]{\strut{}$-1$}}%
      \csname LTb\endcsname%%
      \put(973,1761){\makebox(0,0)[r]{\strut{}$0$}}%
      \csname LTb\endcsname%%
      \put(973,2282){\makebox(0,0)[r]{\strut{}$1$}}%
    }%
    \gplgaddtomacro\gplfronttext{%
      \csname LTb\endcsname%%
      \put(3531,2804){\makebox(0,0)[r]{\strut{}Trajectory}}%
      \csname LTb\endcsname%%
      \put(3531,2565){\makebox(0,0)[r]{\strut{}\textsc{fone}}}%
      \csname LTb\endcsname%%
      \put(1285,496){\makebox(0,0){\strut{}$x_{11}$}}%
      \csname LTb\endcsname%%
      \put(3538,798){\makebox(0,0){\strut{}$x_{21}$}}%
      \csname LTb\endcsname%%
      \put(460,1761){\makebox(0,0){\strut{}$x_{31}$}}%
    }%
    \gplbacktext
    \put(0,0){\includegraphics[width={212.00bp},height={170.00bp}]{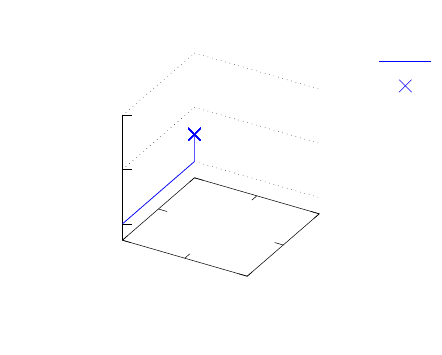}}%
    \gplfronttext
  \end{picture}%
\endgroup

%% file: ex4trajectoryt.tex
% GNUPLOT: LaTeX picture with Postscript
\begingroup
  \makeatletter
  \providecommand\color[2][]{%
    \GenericError{(gnuplot) \space\space\space\@spaces}{%
      Package color not loaded in conjunction with
      terminal option `colourtext'%
    }{See the gnuplot documentation for explanation.%
    }{Either use 'blacktext' in gnuplot or load the package
      color.sty in LaTeX.}%
    \renewcommand\color[2][]{}%
  }%
  \providecommand\includegraphics[2][]{%
    \GenericError{(gnuplot) \space\space\space\@spaces}{%
      Package graphicx or graphics not loaded%
    }{See the gnuplot documentation for explanation.%
    }{The gnuplot epslatex terminal needs graphicx.sty or graphics.sty.}%
    \renewcommand\includegraphics[2][]{}%
  }%
  \providecommand\rotatebox[2]{#2}%
  \@ifundefined{ifGPcolor}{%
    \newif\ifGPcolor
    \GPcolortrue
  }{}%
  \@ifundefined{ifGPblacktext}{%
    \newif\ifGPblacktext
    \GPblacktexttrue
  }{}%
  % define a \g@addto@macro without @ in the name:
  \let\gplgaddtomacro\g@addto@macro
  % define empty templates for all commands taking text:
  \gdef\gplbacktext{}%
  \gdef\gplfronttext{}%
  \makeatother
  \ifGPblacktext
    % no textcolor at all
    \def\colorrgb#1{}%
    \def\colorgray#1{}%
  \else
    % gray or color?
    \ifGPcolor
      \def\colorrgb#1{\color[rgb]{#1}}%
      \def\colorgray#1{\color[gray]{#1}}%
      \expandafter\def\csname LTw\endcsname{\color{white}}%
      \expandafter\def\csname LTb\endcsname{\color{black}}%
      \expandafter\def\csname LTa\endcsname{\color{black}}%
      \expandafter\def\csname LT0\endcsname{\color[rgb]{1,0,0}}%
      \expandafter\def\csname LT1\endcsname{\color[rgb]{0,1,0}}%
      \expandafter\def\csname LT2\endcsname{\color[rgb]{0,0,1}}%
      \expandafter\def\csname LT3\endcsname{\color[rgb]{1,0,1}}%
      \expandafter\def\csname LT4\endcsname{\color[rgb]{0,1,1}}%
      \expandafter\def\csname LT5\endcsname{\color[rgb]{1,1,0}}%
      \expandafter\def\csname LT6\endcsname{\color[rgb]{0,0,0}}%
      \expandafter\def\csname LT7\endcsname{\color[rgb]{1,0.3,0}}%
      \expandafter\def\csname LT8\endcsname{\color[rgb]{0.5,0.5,0.5}}%
    \else
      % gray
      \def\colorrgb#1{\color{black}}%
      \def\colorgray#1{\color[gray]{#1}}%
      \expandafter\def\csname LTw\endcsname{\color{white}}%
      \expandafter\def\csname LTb\endcsname{\color{black}}%
      \expandafter\def\csname LTa\endcsname{\color{black}}%
      \expandafter\def\csname LT0\endcsname{\color{black}}%
      \expandafter\def\csname LT1\endcsname{\color{black}}%
      \expandafter\def\csname LT2\endcsname{\color{black}}%
      \expandafter\def\csname LT3\endcsname{\color{black}}%
      \expandafter\def\csname LT4\endcsname{\color{black}}%
      \expandafter\def\csname LT5\endcsname{\color{black}}%
      \expandafter\def\csname LT6\endcsname{\color{black}}%
      \expandafter\def\csname LT7\endcsname{\color{black}}%
      \expandafter\def\csname LT8\endcsname{\color{black}}%
    \fi
  \fi
    \setlength{\unitlength}{0.0500bp}%
    \ifx\gptboxheight\undefined%
      \newlength{\gptboxheight}%
      \newlength{\gptboxwidth}%
      \newsavebox{\gptboxtext}%
    \fi%
    \setlength{\fboxrule}{0.5pt}%
    \setlength{\fboxsep}{1pt}%
    \definecolor{tbcol}{rgb}{1,1,1}%
\begin{picture}(3960.00,3400.00)%
    \gplgaddtomacro\gplbacktext{%
      \csname LTb\endcsname%%
      \put(793,964){\makebox(0,0)[r]{\strut{}$-1$}}%
      \csname LTb\endcsname%%
      \put(793,1459){\makebox(0,0)[r]{\strut{}$-0.5$}}%
      \csname LTb\endcsname%%
      \put(793,1953){\makebox(0,0)[r]{\strut{}$0$}}%
      \csname LTb\endcsname%%
      \put(793,2448){\makebox(0,0)[r]{\strut{}$0.5$}}%
      \csname LTb\endcsname%%
      \put(793,2942){\makebox(0,0)[r]{\strut{}$1$}}%
      \csname LTb\endcsname%%
      \put(1122,527){\makebox(0,0){\strut{}$-1$}}%
      \csname LTb\endcsname%%
      \put(1694,527){\makebox(0,0){\strut{}$-0.5$}}%
      \csname LTb\endcsname%%
      \put(2266,527){\makebox(0,0){\strut{}$0$}}%
      \csname LTb\endcsname%%
      \put(2837,527){\makebox(0,0){\strut{}$0.5$}}%
      \csname LTb\endcsname%%
      \put(3409,527){\makebox(0,0){\strut{}$1$}}%
    }%
    \gplgaddtomacro\gplfronttext{%
      \csname LTb\endcsname%%
      \put(2001,2924){\makebox(0,0)[r]{\strut{}Trajectory}}%
      \csname LTb\endcsname%%
      \put(2001,2684){\makebox(0,0)[r]{\strut{}\textsc{fone}}}%
      \csname LTb\endcsname%%
      \put(195,1953){\rotatebox{-270.00}{\makebox(0,0){\strut{}$x_{21}$}}}%
      \csname LTb\endcsname%%
      \put(2266,167){\makebox(0,0){\strut{}$x_{11}$}}%
    }%
    \gplbacktext
    \put(0,0){\includegraphics[width={198.00bp},height={170.00bp}]{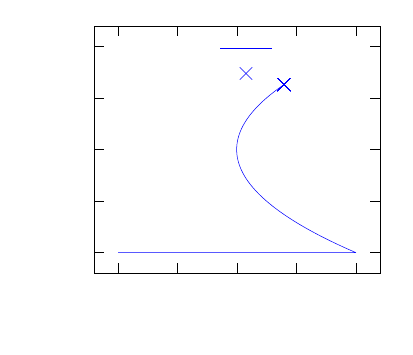}}%
    \gplfronttext
  \end{picture}%
\endgroup

%% file: ex5trajectoryt.tex
% GNUPLOT: LaTeX picture with Postscript
\begingroup
  \makeatletter
  \providecommand\color[2][]{%
    \GenericError{(gnuplot) \space\space\space\@spaces}{%
      Package color not loaded in conjunction with
      terminal option `colourtext'%
    }{See the gnuplot documentation for explanation.%
    }{Either use 'blacktext' in gnuplot or load the package
      color.sty in LaTeX.}%
    \renewcommand\color[2][]{}%
  }%
  \providecommand\includegraphics[2][]{%
    \GenericError{(gnuplot) \space\space\space\@spaces}{%
      Package graphicx or graphics not loaded%
    }{See the gnuplot documentation for explanation.%
    }{The gnuplot epslatex terminal needs graphicx.sty or graphics.sty.}%
    \renewcommand\includegraphics[2][]{}%
  }%
  \providecommand\rotatebox[2]{#2}%
  \@ifundefined{ifGPcolor}{%
    \newif\ifGPcolor
    \GPcolortrue
  }{}%
  \@ifundefined{ifGPblacktext}{%
    \newif\ifGPblacktext
    \GPblacktexttrue
  }{}%
  % define a \g@addto@macro without @ in the name:
  \let\gplgaddtomacro\g@addto@macro
  % define empty templates for all commands taking text:
  \gdef\gplbacktext{}%
  \gdef\gplfronttext{}%
  \makeatother
  \ifGPblacktext
    % no textcolor at all
    \def\colorrgb#1{}%
    \def\colorgray#1{}%
  \else
    % gray or color?
    \ifGPcolor
      \def\colorrgb#1{\color[rgb]{#1}}%
      \def\colorgray#1{\color[gray]{#1}}%
      \expandafter\def\csname LTw\endcsname{\color{white}}%
      \expandafter\def\csname LTb\endcsname{\color{black}}%
      \expandafter\def\csname LTa\endcsname{\color{black}}%
      \expandafter\def\csname LT0\endcsname{\color[rgb]{1,0,0}}%
      \expandafter\def\csname LT1\endcsname{\color[rgb]{0,1,0}}%
      \expandafter\def\csname LT2\endcsname{\color[rgb]{0,0,1}}%
      \expandafter\def\csname LT3\endcsname{\color[rgb]{1,0,1}}%
      \expandafter\def\csname LT4\endcsname{\color[rgb]{0,1,1}}%
      \expandafter\def\csname LT5\endcsname{\color[rgb]{1,1,0}}%
      \expandafter\def\csname LT6\endcsname{\color[rgb]{0,0,0}}%
      \expandafter\def\csname LT7\endcsname{\color[rgb]{1,0.3,0}}%
      \expandafter\def\csname LT8\endcsname{\color[rgb]{0.5,0.5,0.5}}%
    \else
      % gray
      \def\colorrgb#1{\color{black}}%
      \def\colorgray#1{\color[gray]{#1}}%
      \expandafter\def\csname LTw\endcsname{\color{white}}%
      \expandafter\def\csname LTb\endcsname{\color{black}}%
      \expandafter\def\csname LTa\endcsname{\color{black}}%
      \expandafter\def\csname LT0\endcsname{\color{black}}%
      \expandafter\def\csname LT1\endcsname{\color{black}}%
      \expandafter\def\csname LT2\endcsname{\color{black}}%
      \expandafter\def\csname LT3\endcsname{\color{black}}%
      \expandafter\def\csname LT4\endcsname{\color{black}}%
      \expandafter\def\csname LT5\endcsname{\color{black}}%
      \expandafter\def\csname LT6\endcsname{\color{black}}%
      \expandafter\def\csname LT7\endcsname{\color{black}}%
      \expandafter\def\csname LT8\endcsname{\color{black}}%
    \fi
  \fi
    \setlength{\unitlength}{0.0500bp}%
    \ifx\gptboxheight\undefined%
      \newlength{\gptboxheight}%
      \newlength{\gptboxwidth}%
      \newsavebox{\gptboxtext}%
    \fi%
    \setlength{\fboxrule}{0.5pt}%
    \setlength{\fboxsep}{1pt}%
    \definecolor{tbcol}{rgb}{1,1,1}%
\begin{picture}(4240.00,3400.00)%
    \gplgaddtomacro\gplbacktext{%
      \csname LTb\endcsname%%
      \put(1114,980){\makebox(0,0){\strut{}$0$}}%
      \csname LTb\endcsname%%
      \put(2312,634){\makebox(0,0){\strut{}$1$}}%
      \csname LTb\endcsname%%
      \put(2450,678){\makebox(0,0){\strut{}$0$}}%
      \csname LTb\endcsname%%
      \put(3142,1277){\makebox(0,0){\strut{}$1$}}%
      \csname LTb\endcsname%%
      \put(973,1193){\makebox(0,0)[r]{\strut{}$0$}}%
      \csname LTb\endcsname%%
      \put(973,2282){\makebox(0,0)[r]{\strut{}$1$}}%
    }%
    \gplgaddtomacro\gplfronttext{%
      \csname LTb\endcsname%%
      \put(3531,2804){\makebox(0,0)[r]{\strut{}Trajectory}}%
      \csname LTb\endcsname%%
      \put(3531,2565){\makebox(0,0)[r]{\strut{}\textsc{fone}}}%
      \csname LTb\endcsname%%
      \put(1285,496){\makebox(0,0){\strut{}$x_{11}$}}%
      \csname LTb\endcsname%%
      \put(3538,798){\makebox(0,0){\strut{}$x_{21}$}}%
      \csname LTb\endcsname%%
      \put(460,1737){\makebox(0,0){\strut{}$x_{31}$}}%
    }%
    \gplbacktext
    \put(0,0){\includegraphics[width={212.00bp},height={170.00bp}]{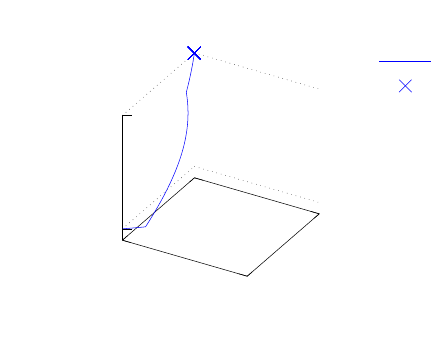}}%
    \gplfronttext
  \end{picture}%
\endgroup

%% file: hypotrajectoryt.tex
% GNUPLOT: LaTeX picture with Postscript
\begingroup
  \makeatletter
  \providecommand\color[2][]{%
    \GenericError{(gnuplot) \space\space\space\@spaces}{%
      Package color not loaded in conjunction with
      terminal option `colourtext'%
    }{See the gnuplot documentation for explanation.%
    }{Either use 'blacktext' in gnuplot or load the package
      color.sty in LaTeX.}%
    \renewcommand\color[2][]{}%
  }%
  \providecommand\includegraphics[2][]{%
    \GenericError{(gnuplot) \space\space\space\@spaces}{%
      Package graphicx or graphics not loaded%
    }{See the gnuplot documentation for explanation.%
    }{The gnuplot epslatex terminal needs graphicx.sty or graphics.sty.}%
    \renewcommand\includegraphics[2][]{}%
  }%
  \providecommand\rotatebox[2]{#2}%
  \@ifundefined{ifGPcolor}{%
    \newif\ifGPcolor
    \GPcolortrue
  }{}%
  \@ifundefined{ifGPblacktext}{%
    \newif\ifGPblacktext
    \GPblacktexttrue
  }{}%
  % define a \g@addto@macro without @ in the name:
  \let\gplgaddtomacro\g@addto@macro
  % define empty templates for all commands taking text:
  \gdef\gplbacktext{}%
  \gdef\gplfronttext{}%
  \makeatother
  \ifGPblacktext
    % no textcolor at all
    \def\colorrgb#1{}%
    \def\colorgray#1{}%
  \else
    % gray or color?
    \ifGPcolor
      \def\colorrgb#1{\color[rgb]{#1}}%
      \def\colorgray#1{\color[gray]{#1}}%
      \expandafter\def\csname LTw\endcsname{\color{white}}%
      \expandafter\def\csname LTb\endcsname{\color{black}}%
      \expandafter\def\csname LTa\endcsname{\color{black}}%
      \expandafter\def\csname LT0\endcsname{\color[rgb]{1,0,0}}%
      \expandafter\def\csname LT1\endcsname{\color[rgb]{0,1,0}}%
      \expandafter\def\csname LT2\endcsname{\color[rgb]{0,0,1}}%
      \expandafter\def\csname LT3\endcsname{\color[rgb]{1,0,1}}%
      \expandafter\def\csname LT4\endcsname{\color[rgb]{0,1,1}}%
      \expandafter\def\csname LT5\endcsname{\color[rgb]{1,1,0}}%
      \expandafter\def\csname LT6\endcsname{\color[rgb]{0,0,0}}%
      \expandafter\def\csname LT7\endcsname{\color[rgb]{1,0.3,0}}%
      \expandafter\def\csname LT8\endcsname{\color[rgb]{0.5,0.5,0.5}}%
    \else
      % gray
      \def\colorrgb#1{\color{black}}%
      \def\colorgray#1{\color[gray]{#1}}%
      \expandafter\def\csname LTw\endcsname{\color{white}}%
      \expandafter\def\csname LTb\endcsname{\color{black}}%
      \expandafter\def\csname LTa\endcsname{\color{black}}%
      \expandafter\def\csname LT0\endcsname{\color{black}}%
      \expandafter\def\csname LT1\endcsname{\color{black}}%
      \expandafter\def\csname LT2\endcsname{\color{black}}%
      \expandafter\def\csname LT3\endcsname{\color{black}}%
      \expandafter\def\csname LT4\endcsname{\color{black}}%
      \expandafter\def\csname LT5\endcsname{\color{black}}%
      \expandafter\def\csname LT6\endcsname{\color{black}}%
      \expandafter\def\csname LT7\endcsname{\color{black}}%
      \expandafter\def\csname LT8\endcsname{\color{black}}%
    \fi
  \fi
    \setlength{\unitlength}{0.0500bp}%
    \ifx\gptboxheight\undefined%
      \newlength{\gptboxheight}%
      \newlength{\gptboxwidth}%
      \newsavebox{\gptboxtext}%
    \fi%
    \setlength{\fboxrule}{0.5pt}%
    \setlength{\fboxsep}{1pt}%
    \definecolor{tbcol}{rgb}{1,1,1}%
\begin{picture}(4520.00,3400.00)%
    \gplgaddtomacro\gplbacktext{%
      \csname LTb\endcsname%%
      \put(1354,951){\makebox(0,0){\strut{}$0$}}%
      \csname LTb\endcsname%%
      \put(2352,663){\makebox(0,0){\strut{}$1$}}%
      \csname LTb\endcsname%%
      \put(2648,728){\makebox(0,0){\strut{}$0$}}%
      \csname LTb\endcsname%%
      \put(3224,1227){\makebox(0,0){\strut{}$1$}}%
      \csname LTb\endcsname%%
      \put(1113,1283){\makebox(0,0)[r]{\strut{}$0$}}%
      \csname LTb\endcsname%%
      \put(1113,2191){\makebox(0,0)[r]{\strut{}$1$}}%
    }%
    \gplgaddtomacro\gplfronttext{%
      \csname LTb\endcsname%%
      \put(3873,2804){\makebox(0,0)[r]{\strut{}$\xi=0.2$}}%
      \csname LTb\endcsname%%
      \put(3873,2565){\makebox(0,0)[r]{\strut{}$\xi=0.8$}}%
      \csname LTb\endcsname%%
      \put(3873,2325){\makebox(0,0)[r]{\strut{}$\xi=3.2$}}%
      \csname LTb\endcsname%%
      \put(1425,496){\makebox(0,0){\strut{}$\phi_0$}}%
      \csname LTb\endcsname%%
      \put(3678,798){\makebox(0,0){\strut{}$\phi_1$}}%
      \csname LTb\endcsname%%
      \put(600,1737){\makebox(0,0){\strut{}$q$}}%
    }%
    \gplbacktext
    \put(0,0){\includegraphics[width={226.00bp},height={170.00bp}]{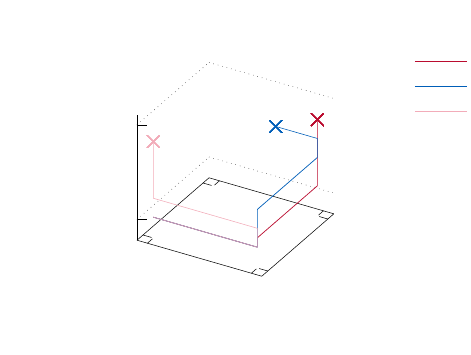}}%
    \gplfronttext
  \end{picture}%
\endgroup